\documentclass[11pt]{amsart}
\usepackage{fullpage}
\usepackage{mydefs}
\usepackage{amssymb}
\usepackage{graphicx}

\usepackage{amsmath}

\usepackage{graphicx}
\usepackage{amssymb}
\usepackage{multicol}
\usepackage{epstopdf}
\usepackage{mydefs}
\usepackage{algorithm}
\usepackage{algorithmic}

\usepackage{url}

\newcommand{\alg}[1]{\textbf{#1}}   %% Notation for algorithm names

\newcommand{\calA}{\mathcal{A}}      

\def\mainC{{\rho}}
\def\avgA{{\bar{A}}}
\def\effC{{5 \rho}}

\title{Wireless Network Stability in the SINR Model}
\author[E. I.\'Asgeirsson]{Eyj\'olfur I. \'Asgeirsson}
\address[E. I.\'Asgeirsson]{School of Science and Engineering\\
 Reykjavik University\\
 101 Reykjavik, Iceland\\}
\email{eyjo@ru.is}

\author[M. M. Halld\'orsson]{Magn\'us M. Halld\'orsson}
\address[M. M. Halld\'orsson]{School of Computer Science\\
 Reykjavik University\\
 101 Reykjavik, Iceland\\}
\email{mmh@ru.is}

\author[P. Mitra]{Pradipta Mitra}
\address[P. Mitra]{School of Computer Science\\
Reykjavik University\\
Reykjavik 101, Iceland}
\email{ppmitra@gmail.com}

\begin{document}

\begin{abstract}
We study link scheduling in wireless networks under stochastic arrival processes
of packets, and give an algorithm that achieves stability in the
physical (SINR) interference model. The efficiency of such an algorithm is the
fraction of the maximum feasible traffic that the algorithm can handle
without queues growing indefinitely.  
Our algorithm achieves two important goals:
(i) efficiency is independent of the size of the network, and (ii) the
algorithm is fully distributed, i.e., individual nodes need no
information about the overall network topology, not even local
information.
%
\iffalse
Specifically, we make the following contributions. We design a distributed algorithm that achieves 
$\Omega(\frac{1}{\log^2 n})$-efficiency on all networks (where $n$ is the number of links in the network), for all length-monotone, sublinear
power assignments. In the same setup, we design an alternative distributed algorithm with $\Omega(\log \Delta)$
efficiency (where $\Delta$ is the length diversity of the network).
We design a fast centralized algorithm 
with  $\Omega(\frac{1}{\log n})$-efficiency for the power control problem, and a distributed algorithm for the same problem with $\Omega(\frac{1}{\log n(\log n + \log \log \Delta)})$-efficiency (where $\Delta$ is the length diversity of the link set).
%Our goal is to study the stability of these algorithms
%when packets on the links appear according to stochastic arrival processes.
%We propose a number of algorithms for different settings,
%and characterize the arrival processes for which our algorithms
%stabilize the networks. 
%Related lower bounds and numerical experiments are provided.
\fi
\end{abstract}

\maketitle 
\section{Introduction}
% Intro
We study the problem of scheduling packets over links in a wireless network,
each link being a sender-receiver pair of wireless nodes.
%A ``link'' is a ordered pair of two wireless devices, denoting a potential or actual transmission from the first (the ``sender'') to the second (the ``receiver'').
%
Since wireless signals propagate in all directions, simultaneous
transmissions interfere with each other. This interference limits the
number of transmissions that can succeed simultaneously. A wireless
packet scheduling algorithm thus has to schedule packets efficiently,
with respect to the limits imposed by interference.

In this setting, consider the following two related problems. The first: Given a set of links, how quickly (i.e., using how few slots) can all of the links be scheduled, taking interference into account? The second: Given is a set of links, and packets arrive at the senders of each link according to some stochastic process, where they remain queued until successfully transmitted to the receiver. Can one ensure that queue sizes at the senders remain bounded, in expectation? The first question is an ``off-line'' algorithmic problem, whereas the second comes from a queueing theoretic perspective where the input is probabilistic.
In spite of their obvious commonalities, they are generally studied using quite disparate techniques. Our goal in this paper is to bridge a gap between these two related areas; specifically, to use recently developed algorithmic techniques to achieve results for the stochastic setting.

%% Models of interference
To do this, a crucial first step is to choose the right interference model --
one that is faithful to physical reality yet is simple enough to be rigorously analyzable. In this paper, we adopt the SINR (Signal to Interference and Noise ratio) or \emph{physical} model of interference. Compared to the more traditional and widely studied graph based models, the SINR model has been found to be realistic, and is enjoying
increased attention and adoption~\cite{MaheshwariJD08,Moscibroda2006Protocol,moscibroda06b}. This model (precisely defined in Section \ref{sec:model}) is based on a realistic geometry of signal propagation (compared to unrealistic graph based interference models).

% Stability
We are thus interested in algorithms that keep queue sizes bounded
when faced with stochastic packet arrivals over arbitrary periods of
time, assuming the SINR interference model.  A network in which this
goal is achieved is called \emph{stable}.  Stability has been an
widely-studied metric for analyzing the performance of scheduling
algorithms for wireless networks for quite some time.  In a seminal
work, Tassiulas and Ephremides \cite{TE92}, gave a characterization of
those stochastic processes for which stability is possible in
principle. The characterization is general enough to work for
virtually any interference model.  In light of this, the goal for the
algorithm designer is to produce an (simple and efficient) algorithm
that stabilizes networks for all (or a fair chunk) of these arrival
processes.  There is a long tradition of such work, e.g.,
\cite{bestInfocom08,DBLP:conf/sigmetrics/ModianoSZ06},
but they almost exclusively apply to graph-based interference models.

% SINR & Stability
The Tassiulas-Ephremides characterization \emph{is} formulated as a computational problem, which, if solved, would stabilize a network under all potentially stabilizable arrival rates.
However, for the SINR model, this problem (known as the maximum
weighted \emph{capacity} problem) is NP-hard \cite{DBLP:conf/infocom/AndrewsD09}. Not much is known about the algorithmic complexity of this problem (see \cite{infocom12} about a recent centralized result for linear power and more discussion about its relation to network stability), and almost nothing about possible distributed implementations. Thus, alternative approaches need to be sought.

In this work, we develop efficient and distributed
scheduling algorithms for wireless network stability ---
by applying intuitions developed in recent 
research on the SINR model (\cite{SODA11} and \cite{icalp11} contain many references), all of which provide approximation
algorithms to some relevant algorithmic questions.

One possible approach would be to apply algorithms for some of the
core optimization problems as black boxes. For example, there are
constant factor approximation algorithms for the capacity problem
\cite{HW09,SODA11,KesselheimSoda11}. These alone are not
sufficient, as there is no guarantee of fairness.  Still, they can be
easily turned into a $O(\log n)$-approximation for the weighted
capacity problem.  The problem with this approach, however,
is that these algorithms are centralized, with no effective
distributed algorithms in sight.  Distributed algorithms are
 of crucial importance in the current setting.  Our
approach is therefore more of a ``gray-box'' one -- while we adopt an
algorithm of \cite{KV10} as our basis, our analysis depends not on the
overall approximation factor, but on more subtle properties of that
algorithm.

% Main issues: Efficiency and distributed implementability
Apart from being distributed, the main property a 
scheduling algorithm should have is high efficiency.
``Efficiency'' has a specific technical meaning which we define in Section \ref{sec:model}.
Intuitively, it captures
how well the algorithm does compared to the Tassiulas-Ephremides characterization.

% The efficiency should be either independent, or barring that, only mildly dependent on 
% network parameters, the relevant 
% parameters for the SINR model being the size of the network, $n$ and
% the ratio of the largest and the smallest links, $\Delta$. 

% Our results: Efficiency
We achieve, depending on the algorithm chosen, efficiency ratios of
$\Omega(\frac{1}{\log^2 n})$, $\Omega(\frac{1}{\log n})$ and
$\Omega(\frac{1}{\log n (\log n + \log\log \Delta)})$, that are
comparable or better than existing work on this topic.  Our main
algorithm requires only a ``carrier sense'' primitive to make it
completely distributed.  This is in contrast to many distributed
algorithms in the literature (e.g.,
\cite{DBLP:conf/sigmetrics/ModianoSZ06,lqfmobihoc}) that are better
described as ``localized'' -- requiring an underlying infrastructure
for wireless nodes to communicate with nearby nodes. 
This infrastructure, moreover, is usually not subject to the
interference constraints of the original network.  This is a rather
strong assumption, especially in light of the fact that in a wireless
network, one is presumably trying to establish such an infrastructure
in the first place.

% Organization of the paper
The paper is organized as follows.  In Sections \ref{sec:model} and
\ref{sec:results} we present the system model, our results, and more
specific discussion on related work.  In Section \ref{sublin} we
describe a general algorithmic framework for wireless scheduling. We
then provide a specific instantiation of this framework that achieves
good throughput performance for a large class of power assignments in
general metric spaces, with implications for the power control case, where the power
can be selected by links separately. Finally, in Section \ref{sec:pc}, we
prove a more efficient, but centralized result for the power control case and present simulation results.

\section{Model and Preliminaries}
\label{sec:model}
\subsubsection*{The SINR Model}
The wireless network is modeled as a set $L$ of $n$ links, where
each link $l \in L$ represents a potential transmission from a sender
$s_l$ to a receiver $r_l$, both points in a metric space. 
The distance between two points $x$ and $y$ is denoted $d(x,y)$.  
The distance from
$l'$'s sender to $l$'s receiver is denoted $d_{l'l} = d(s_{l'}, r_{l})$.
The length of link $l$ is denoted 
simply by $\ell = d(s_l, r_l)$.

The set may be associated with a \emph{power assignment}, which is an assignment of a transmission
power $P_l$ to be used by each link $l \in L$.
The signal received at point
$y$ from a sender at point $x$ with power $P$  is $P/d(x, y)^\alpha$ where the constant 
$\alpha > 0$ is the \emph{path-loss exponent}. 

We can now describe the  \emph{physical} or SINR-model of interference. In this model, a receiver $r_l$
successfully receives a message from the sender $s_l$ if and only if the
following condition holds:
\begin{equation}
 \frac{P_l/\ell^\alpha}{\sum_{l' \in S \setminus  \{l\}} P_{l'}/d_{l'l}^\alpha + N} \ge \beta \ , 
 \label{eq:sinr}
\end{equation}
where $N$ is the environmental noise, the constant $\beta \ge 1$ denotes the minimum
SINR (signal-to-interference-noise-ratio) required for a message to be successfully received,
and $S$ is the set of concurrently scheduled links in the same \emph{slot} (we assume that time is slotted.).
We say that $S$ is \emph{SINR-feasible} (or simply \emph{feasible}) if (\ref{eq:sinr}) is
satisfied for each link in $S$.

A power assignment $P$ is \emph{length-monotone} if $P_v \ge P_w$
whenever $\ell_v \ge \ell_w$ and \emph{sub-linear} if
$\frac{P_v}{\ell_v^{\alpha}} \le \frac{P_w}{\ell_w^{\alpha}}$ whenever
$\ell_v \ge \ell_w$. This class includes the most interesting and
practical power assignments, such that uniform power (all links use
the same power), linear power ($P_l = \ell^{\alpha}$, known to be
energy efficient in the presence of noise), and mean power ($P_l =
\ell^{\alpha/2}$, the assignment that produces maximum capacity in
this class). We will also consider the ``power control'' case, where
the power assignments are not predetermined, but have to be found out
by the algorithm, and can be arbitrary.

% We assume that $|L| = n$.  
Let $\Delta = \frac{\ell_{\max}}{\ell_{\min}}$, where $\ell_{\max}$ and $\ell_{\min}$ are, respectively, the maximum and minimum lengths in $L$.

\begin{defn}
The \textbf{affectance} $a^P_{l'}(l)$ of link $l$ caused by another link $l'$,
with a given power assignment $P$,
is the interference of $l'$ on $l$ relative to the power
received, or
  \[ a^P_{l'}(l) 
     = \min\left\{1, c_v \frac{P_{l'}}{P_l} \cdot
     \left(\frac{\ell}{d_{l' l}}\right)^\alpha\right\}\ ,
  \] 
where $c_v = \beta/(1 - \beta N \ell^\alpha/P_l)$. 
\end{defn}

The definition of affectance was introduced in \cite{GHWW09} and achieved the form
we are using in \cite{KV10}.
When clear from the context we drop the superscript $P$. 
Also, let $a^P_l(l) = 0$.
Using the idea of affectance, Eqn.~\ref{eq:sinr} can be rewritten as 
\[ a^P_S(l) \equiv \sum_{l' \in S}a^P_{l'}(l) \leq 1\ , \]
for all $l \in S$.

 A link can schedule at most one packet during a slot, in other words,
if a link has a queue, at most one packet from the queue can be scheduled during a single slot.

\subsubsection*{Stability of Stochastic Processes}

We assume that packets arrive at the sender of each link $l$ according to a stochastic
process with average arrival rate $m_l$.

We define stability as such.
\begin{defn}
An algorithm \textbf{stabilizes} a network for a particular arrival process if, under that
arrival process the average queue size at each link is bounded (ie, does not grow asymptotically with time).
\end{defn}
The \emph{throughput region} is then the set of all possible arrival rate
vectors such that there exists some scheduling policy that can stabilize the
network.
As proved in \cite{TE92}, the throughput region is characterized by
\[ \Lambda = \{\lambda : \lambda \preceq \phi, \text{for some } \phi
   \in Co(\Omega)\}\ , \]
here $\Omega$ is the set of vectors in $\mathbb{R}^n$ characterizing 
all maximal feasible sets (i.e., each vector in $\Omega$ is a binary vector with
1's in indices corresnponding to links belonging to the relevant maximal feasible set). 
$Co(\Omega)$ is the convex hull of $\Omega$; $\lambda$ and $\phi$ are vectors in $\mathbb{R}^n$, indicating arrival rates on links, and $\lambda \preceq \phi$ means that each entry of $\lambda$ is less than or equal to the corresponding entry in $\phi$.

In the best case, one would like stabilize all of $\Lambda$. If that is not possible, the hope is to achieve a high \emph{efficiency ratio}:

\begin{defn}
The \textbf{efficiency ratio} $\gamma$ of a scheduling algorithm 
is $\gamma = \sup\{\gamma: \text{all networks are}$ stabilized for all
$\lambda \in \gamma \Lambda \}$. 
The algorithm is then \textbf{$\gamma$-efficient}.
\end{defn}
We assume that the arrival processes are independent accross time (and links).
For a certain efficiency $\gamma$, for all permissible arrival rate vectors $\lambda$, $ \lambda \preceq \sum_{i} m_i M_i$ where each $M_i$ is a maximal feasible set, and $m_i$ are weights such that $\sum_i m_i  = \gamma$. Let the
expected arrival rate on a link $l$ be $m_l$, it can be easily seen that 
\begin{equation}
\label{linktofs}
m_l = \sum_{i: l \in M_i} m_i \leq \gamma \ .
\end{equation}

\section{Results and related work}
\label{sec:results}
Our main results are:

\begin{theorem}
\label{mainth1}
For all given networks with links on metric spaces, and all sub-linear, length-monotone 
power assignments, there exists a $\Omega(\frac{1}{\log^2 n})$-efficient distributed algorithm.
% Additionally, the expected waiting time for a packet is $O(\log^2 n)$.
\end{theorem}
%This is given in Sec.~\ref{sublin}.
%\footnote{How shall we mention the hardness result in Section 5? If we  don't mention it here, we should make it into a subsection. --- made it into subsection.}

%\noindent We treat the power control problem in Sec.~\ref{sec:pc}.

\begin{theorem}
\label{mainth2}
For all given networks with links on the Euclidean plane,
there exists a $\Omega\left(\frac{1}{\log n (\log n + \log \log \Delta)}\right)$-efficient distributed power control algorithm.  A 
centralized algorithm exists that achieves
$\Omega\left(\frac{1}{\log n}\right)$ efficiency.
\end{theorem}
%The algorithm in the above network is not distributed. We will show as a corollary of 
%Thm.~\ref{mainthm} a

%We remark that Thm.~\ref{mainth2} can be generalized to doubling metrics where $\alpha$ is strictly greater than the ``doubling constant'' (known as \emph{fading metrics} \cite{DBLP:conf/esa/Halldorsson09}). 

We are aware of two earlier papers on stability in the SINR model.
In \cite{lqfmobihoc}, the authors study the Longest Queue First algorithm (a classical algorithm that can be seen as a natural extension of maximal weighted matching). They show that LQF is not stable, but a variant works well.
A ``localized'' implementation is provided, i.e., it is shown that the algorithm can be implemented in a distributed manner if links can communicate with
other ``neighboring'' links arbitrarily. The achieved efficiency ratio in
$\Omega(\frac1{\Delta^{\alpha}})$ (the dependence on $\Delta$ is not explicitly mentioned, but can be seen to be necessary). Our recent paper \cite{ciss12} is a companion of the current work, 
where an extremely simple and completely distributed algorithm achieving $\Omega(\frac1{\Delta^{\alpha}})$ efficiency is introduced.
In comparison, the results in the current work involve efficiency that is logarithmically dependent on $n$ (and, in one case, doubly logarithmic in $\Delta$). 
Dependence on $n$ and $\Delta$ are theoretically not comparable, and either could be preferable in practice. The distributed algorithm in this paper has to assume
a carrier-sense primitive, which is not assumed in \cite{ciss12}. However, it does not need to have a special communication infrastructure with neighboring links.

The body of work on wireless network stability in \emph{other} models is too vast to survey properly. In terms of efficiency ratio, 
a range of results have been derived
in a variety of models. Naturally one seeks efficiency of $1$ \cite{DBLP:conf/sigmetrics/ModianoSZ06} whenever
possible, but results for efficiency ratios of $1$ under certain conditions \cite{secordorder}, or $\frac{1}{6}$ \cite{bestInfocom08} can be found in the literature.
Ratios in terms of certain network characteristics are known as well \cite{bestInfocom08,lqfmobihoc}. 
For the SINR model, which is being studied only very recently, an efficiency ratio of a
constant that is independent of network parameters is not known. 

Technically, we depend heavily on \cite{KV10} that provides a $O(\log^2 n)$ approximation algorithm for the scheduling problem. The algorithm and technical aspects of this work used here will be introduced in the following section as needed.

We are aware of a very recent unpublished work of Kesselheim \cite{dynamickesselheim} achieving results in the SINR model very similar to the present paper.

\section{Main Algorithm}
\label{sublin}

The basic algorithmic framework used is listed as \alg{General} below. For simplicity,
we treat it as a centralized procedure first and discuss distributed implementations later.

\begin{algorithm}                      % enter the algorithm environment
\caption{General($\theta$, $\calA$)}          % give the algorithm a caption
                           % and a label for \ref{} commands later in the document
\begin{algorithmic}[1]                    % enter the algorithmic environment
%\LinesNumbered
%\REQUIRE $n \geq 0 \vee x \neq 0$
     \STATE The algorithm maintains a FIFO queue $\mathcal{S}$ of sets, such that each $S\in \mathcal{S}$ is feasible. 
     \STATE At the beginning, $\mathcal{S} \leftarrow \emptyset$.
     \FOR{time $t \leftarrow{} 1, 2, \ldots $}
       \IF{$\mathcal{S}$ is non-empty}
         \STATE Schedule the first $S \in \mathcal{S}$         
         \STATE $\mathcal{S}\leftarrow{}\mathcal{S} - S$
       \ENDIF
       \IF{$t{}\bmod{}\theta = 1$}
	\STATE Let $L = $ new packet arrivals in the time period $[t - \theta, t - 1]$
	\STATE $q = t / \theta$
	\STATE Use algorithm $\calA$ to find 	   a schedule $\mathcal{R}_q = \cup_j R'_j$ for $L$
	\STATE Append $\mathcal{R}_q$ (in any order) to $\mathcal{S}$
	\ENDIF
       \ENDFOR
\end{algorithmic}
\label{alg1fig}
\end{algorithm}

The algorithm takes two parameters. One is $\theta$, a number that defines the ``period''
of the algorithm. The second parameter is  an algorithm $\calA$ which can
solve the \emph{scheduling problem}, used
as a black box by \alg{General} to compute schedules. The \emph{scheduling problem} is the optimization
problem where given a set $L$ of links, one seeks to partition $L$ into minimum number of sets such
that each of these sets is feasible (i.e., can be transmitted in one slot). Since the problem is NP-hard,
we will work with approximation algorithms. Depending on the result we seek, we will set $\theta$ and $\calA$
accordingly.

\alg{General} can be alternatively described in the following way. The algorithm divides 
the time slots into consecutive
periods of length $\theta$ each. Let us denote these periods as $C_1, C_2 \ldots$ etc.
At the beginning of period $C_q$, the algorithm
computes a schedule $\mathcal{R}_{q-1}$ of the links produced in $C_{q-1}$. It does so using
$\calA$. \alg{General} then adds these computed feasible sets to the set $\mathcal{S}$. 
Now during each slot of $C_q$, the algorithm schedules
the first set from $\mathcal{S}$ (which is implemented as a FIFO queue). 
It does this until $C_q$ ends, in which case it moves on to the next period, or
until  $\mathcal{S}$ is empty, in which case it waits until the end of $C_{q}$. Note that there is
nothing to schedule during $C_1$, we just wait during this time. 

Let $Q^t_l$ be the queue length at link $l$ at time $t$.
First, note that $Q^t_l \leq S^t$ for all $l$, where $S^t =
|\mathcal{S}|$ at time $t$ ($\mathcal{S}$ is as in the algorithm $\alg{General}$). This is the
number of slots we require  to schedule all links outstanding at time $t$. Obviously, one cannot schedule all links
in time less than the size of the longest queue (since copies of the same link cannot be scheduled together).
Thus a bound on $S^t$ immediately gives us a bound on $Q^t_l$ (for all $l$). Consequently, from now on we will focus on
bounding $S^t$. Also note that it suffices to bound $S^t$ on period boundaries, i.e., at times $t$ such that $t \bmod \theta = 1$. This is because, in expectation, the queue lengths cannot grow by much during the course of a period.
Let $\mainC$ be a large enough constant.

For simplicity, we will assume that the arrival distributions at every link $l$ is a Bernoulli random variable with mean $m_l$. 
To prove Theorem \ref{mainth1}, we set $\theta = 10 \mainC^2 \log^2 n$. 

\subsection{The scheduling algorithm}
We also select as $\calA$
the scheduling algorithm described in \cite{KV10}. It is known \cite{icalp11} that
this algorithm achieves a O$(\log n)$-approximation factor to the scheduling problem. We are, however, interested in a slightly different performance measure of the algorithm in \cite{KV10}. For a link set $R$, 
define the \emph{maximum average affectance}
$\avgA(R) = \max_{Q \subseteq R} \frac{1}{|Q|}\sum_{l \in Q}\sum_{l' \in Q} a_{l}(l')$. It is known that:

\begin{theorem}{\cite{KV10}}
  The algorithm $\calA$ has expected running time of at most $\mainC \log n \cdot \avgA(R)$ on a link set $R$.
  \label{algeff}
\end{theorem}

We now turn our attention to proving the stability of the algorithm, i.e., Thm.~\ref{mainth1}.
Given the efficiency claimed in Thm.~\ref{mainth1}, it sufficient to deal with stochastic processes satisfying
\begin{equation}
\sum_{M_i} m_i \leq \frac1{5 \rho^2 \log^2 n} \ .
\label{loadbound1}
\end{equation}

\begin{lemma}
\label{schednumberbound1}
$\Ex(|\mathcal{R}_q|) \leq O(\log^2 n) < \theta$  for all $q$.
%, where $\mathcal{R}_q$ is as defined in the algorithm.
\end{lemma}
\begin{proof}
%Number of slots of $C_q$ as slot $t = 1, 2 \ldots \theta$.  
Define $a^{t+}(l)$ to be the outgoing affectance from a link $l$ to longer links appearing in slot $t$. Formally, if $X_l(t)$ is the Bernoulli random variable denoting the number of packets that arrived at the sender of link $l$ during slot $t$, then
\begin{equation}
a^{t+}(l) \equiv \sum_{l' \in L, l' \geq l} a_{l}(l') X_{l'}(t) \ .
\label{aplusdef}
\end{equation}
Let $A^+(l)$ be the sum of $a^{t+}(l)$ over all $\theta$ slots in the period, or,
\[ A^+(l) = \sum_{t = 1}^{\theta} a^{t+}(l) \ . \]
Now, we claim,
\begin{claim}
$\Ex(A^+(l)) \leq \frac{\theta}{\effC \log n}$.
\label{expectedA}
\end{claim}
\begin{proof}
In \cite{KV10}, it is shown that for any feasible set $M_i$
\begin{equation}
\label{afftohig}
\sum_{l' \in M_i, l' \geq l} a_l(l') \leq \mainC \log n \ .
\end{equation}

Thus, for any time slot $t$,
\begin{align*}
\Ex(a^{t+}(l)) & \overset{1}{=} \sum_{l' \in L, l' \geq l} a_{l}(l') \Ex(X_{l'}(t)) \overset{2}{=} \sum_{l' \geq l} a_l(l') m_{l'} \overset{3}{\le} \sum_{l' \geq l} a_l(l') \sum_{i: l' \in M_i} m_i \\
& \overset{4}{=} \sum_{i} m_i \sum_{l' \in M_i, l' \geq l} a_l(l') 
 \overset{5}{\le} \mainC \log n \sum_i m_i \leq \frac{1}{\effC \log n} \ ,
\end{align*}
 with explanations of numbered (in)equalities:
\begin{enumerate}
\item By definition of $a^{t+}(l)$ (Eqn.~\ref{aplusdef})
\item $\Ex(X_{l'}(t)) = m_{l'}$ by definition, since they both express the expected number of packets arriving in each time slot on $l'$.
\item By Eqn.~\ref{linktofs}.
\item Rearrangement.
\item By Eqn. \ref{afftohig}.
\item By Eqn.~\ref{loadbound1}.
\end{enumerate}
\end{proof}

Now, by the Chernoff-Hoeffding inequality (see, for example, Eqn.~1.8, Thm.~1.1 of \cite{dubhashi}):
\[ \Pro(A^+(l) \geq r) \leq \frac{1}{2^{r}}\ , \]
for all $r \geq \frac{2 \theta}{\effC \log n} = 4\rho\log n$. Defining $A^+_{\max} = \max_l A^+(l)$, and union bounding we get,

\begin{equation}
\Pro(A^+_{\max} \geq r) \leq \frac{n}{2^{r}} \ .
\end{equation}

\noindent We analogously define 
$a^{t-}(l) \equiv \sum_{l' \in L, l' \geq l} a_{l'}(l) X_{l'}(t)$ to be the incoming affectance from longer links. We similarly define $A^-_{max}$, and obtain that
$\Pro(A^-_{\max} \geq r) \leq \frac{n}{2^{r}}$,
%\begin{equation}
%a^{t-}(l) \equiv \sum_{l' \in L, l' \geq l} a_{l'}(l) X_{l'}(t) \ .
%\end{equation}
%We can now analogously define $A^-_{\max}$, and prove:
%\begin{equation}
%\Pro(A^-_{\max} \geq r) \leq \frac{n}{2^{r}} \ ,
%\end{equation}
which depends on the bound
\begin{equation} 
\label{afffromhig}
\sum_{l' \in M_i, l' \geq l} a_{l'}(l) \leq \mainC \ ,
\end{equation}
also proven in \cite{KV10}. (Note that the bound here is
tighter than Eqn.~\ref{afftohig}).

%For a set of links $L$, define the maximum average affectance 
%$\avgA(L) = \max_{M \subseteq L} \frac{1}{|M|}\sum_{l \in M}\sum_{l' \in M} a_{l}(l')$.
It is not hard to verify that that
$\avgA \leq A^+_{\max} + A^-_{\max}$. Thus,
\begin{align*}
\Ex(\avgA) & \leq \Ex(A^+_{\max}) + \Ex(A^-_{\max}) \leq 2 \Ex(A^+_{\max}) \\
& \leq \frac{2 \theta}{\effC \log n} + \sum_{i = 1}^{\infty} 2^i \frac{2 \theta}{\effC \log n} \frac{n}{2^{2^i \frac{2 \theta}{\effC \log n}}} \\
& \leq  \frac{2 \theta}{\effC \log n}  + \frac{2 \theta}{\effC \log n} 
 = \frac{4 \theta}{\effC \log n} \ .
\end{align*}

Now, by Theorem \ref{algeff}, $\Ex_{\calA}(|\mathcal{R}_q|) \leq \mainC \log n \cdot \avgA$, where $\Ex_{\calA}$ denotes expectation over the random bits of the algorithm (which is randomized).  Noting that the random bits of the algorithm are independent of the arrival process, we can use the bound on $\Ex(\avgA)$
to claim that 

\[ \Ex(|\mathcal{R}_q|) \leq \mainC \log n \cdot \frac{4 \theta}{\effC \log n} < \theta \ .\]
\end{proof}

Now we can prove Thm.~\ref{mainth1}. 
Note that by
Lemma \ref{schednumberbound1}, the expected scheduling cost required for packets produced
during a single period ($\Ex(|R_q|)$) is strictly smaller than the scheduling capacity of a single period ($\theta$).
With this observation, we can reduce our system to a very basic queueing system:

\begin{itemize}
\item A single server, an infinite queue, and slotted time. The time slots in this system correspond to the periods of the original system.
\item At the beginning of each slot a fixed number $s$ of packets are served and leave the system ($s$ corresponds to $\theta$).
\item At the end of every slot, a random number of new packets arrive. This is a random variable $A$ on the  non-negative
integers, and $\Ex(A) = a$ ($a$ corresponds to $\Ex(|\mathcal{R}_q|)$).
\end{itemize}
Now, if $a < m$, the corresponding countable Markov chain has a stationary distribution, and if $a$ is square integrable,
the expected queue length will be finite (see \cite{Asmussen}, for example). The condition $a < m$ is easily seen to be true 
(since $\Ex(|\mathcal{R}_q|) < \theta$ by Lemma \ref{schednumberbound1}). Square integrability of 
$|\mathcal{R}_q|$ follows from the fact that $\mathcal{R}_q$ admits a large deviation bound (this
is implicit in the proof of Lemma \ref{schednumberbound1}). Of course the ``queue length'' of this system 
corresponds to $S^t$, the number of slots the outstanding packets would require to be scheduled. 
Thus we have proven  $\Ex(S^t)$ to be bounded, and as observed before, this is enough to complete the proof of the theorem.

\subsection{Implications for the power control problem}

It was shown
in \cite{DBLP:conf/esa/Halldorsson09,SODA11} that
\emph{mean power} (where $P_l$ is set to $\ell^{\alpha/2}$) achieves a
good approximation to the power control problem. Since
Thm.~\ref{mainth1} covers the mean power assignment, this gives us a
 distributed algorithm for the power control problem. To
achieve the bound claimed in Thm.~\ref{mainth2} for the distributed
algorithm, the main ingredient is the following bound (analogous to
Eqn.~\ref{afftohig}).

\begin{lemma}[\cite{DBLP:conf/esa/Halldorsson09,SODA11}]
If M is a feasible set with respect to \emph{any} power assignment,
\begin{equation} 
\sum_{l' \in M, l' \geq l} a_{l'}(l) \leq c_3 (\log n + \log \log \Delta) \ ,
\end{equation}
where the affectance $a_{l'}(l)$ is measured using mean power.
\end{lemma}
With this bound (and a similar analog of Eqn.~\ref{afffromhig}) in hand, the proof technique of Thm.~\ref{mainth2}
can be duplicated to achieve the bound claimed in Thm.~\ref{mainth2} for distributed algorithms.

\subsection{Distributed implementation}
\label{sec:distr}

We now demonstrate how to implement \alg{General} in a distributed fashion.
Interestingly, we require very few additional assumptions to make this work. 
The basic tool is the 
algorithm in \cite{KV10}, listed as \alg{Distr-SingleLink}. 

\begin{algorithm}                      % enter the algorithm environment
\caption{Distr-SingleLink}          % give the algorithm a caption
\label{alg3}                           % and a label for \ref{} commands later in the document
\begin{algorithmic}[1]                    % enter the algorithmic environment
%\LinesNumbered
%\REQUIRE $n \geq 0 \vee x \neq 0$
     \STATE $k \leftarrow 0$
     \WHILE{transmission not successful}
       \STATE $q =  \frac{1}{4 \cdot 2^k}$
       \FOR{$\frac{8 \ln n}{q}$ slots}
          \STATE transmit with i.i.d. probability $q$
       \ENDFOR
       \STATE $k \leftarrow k + 1$
     \ENDWHILE
\end{algorithmic}
\label{alg13ig}
\end{algorithm}

The first thing to note here is that \alg{Distr-SingleLink} itself is completely distributed. 
In other words, \alg{General}, if applied to the links
produced during a single period, could be implemented in a distributed manner straightaway. 
The challenge is that our bounds assume that the algorithm works in a FIFO
manner, thus \alg{Distr-SingleLink} for a packet produced in $C_q$ should not start
executing until all packets produced up to $C_{q-1}$ have been successfully scheduled. 

To implement this, we assume that each sender in the system maintains a few counters. 
The first counter $cur$ keeps track of the current period, and a second counter $s$ tracks the current
period being \emph{scheduled} (naturally $s \leq cur$). 
For each outstanding packet $p$ in the queue, the sender also maintains the period
in which it was generated ($g_p$). 
The counter $cur$ is easily maintained, by incrementing it once every $\theta$ slots. 
The third counter is equally simple, when a packet $p$ arrives, $g_p$ is assigned
the current value of $cur$. We will describe how $s$ is maintained below, but note that
given $s$, the algorithm can now be easily implemented in a distributed fashion. For each
packet $p$, the sender waits until $s = g_p$, and then runs \alg{Distr-SingleLink} for $p$ until
the link successfully transmits.

Maintaining $s$ is slightly more, but not too, involved. 
It is here that we need to make an additional assumption, which is that nodes (senders) can
sense the channel to determine transmission activity. 
This is a not uncommon assumption (see, e.g., \cite{ScheidelerRS08}), based on the Clear Channel Assessment capability in the 802.11 wireless standard.

Let us divide the time slots into consecutive pairs.  The first time
slot is used for normal transmission (i.e., executing \alg{General}
and \alg{Distr-SingleLink}).  The second slot is used for
signaling. Senders that are transmitting currently (i.e., senders that
have at least one non-transmitted packet $p$ for which $g_p = s$) use
the signaling slot to simply signal that they have still not
completed.  Thus when all links from period $s$ succeed, a silent
signaling slot appears.  All senders register this event by sniffing
the channel, and increment $s$.

From a practical point of view, wasting every other slot for
signaling, as well as assuming some sort of ``perfect'' carrier sensing
capability, is problematic. 
Simulation studies presented are done
with practically reasonable approximations to these
assumptions (indeed, these heuristics seem to help the algorithm).

\section{Extensions and simulations}
\label{sec:pc}
In the power control version of the problem, selecting (an arbitrary) power for
each link is a part of the problem.
Feasible sets are now those for which there exists an (unknown) power
assignment that allows for simultaneous transmissions of all the links.

Using a centralized algorithm of \cite{KesselheimSoda11}, a $O(\frac1{\log n})$ throughput is achievable.
We omit details.

\begin{figure}
\begin{center}
\includegraphics[width=0.49\textwidth]{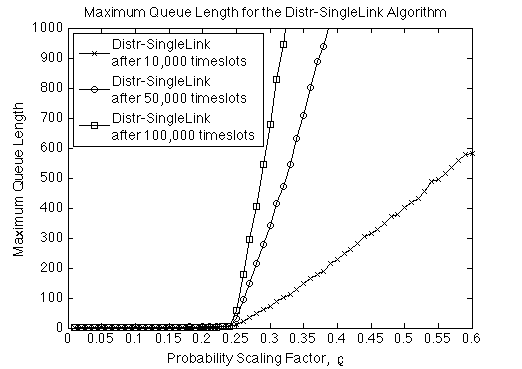}
\includegraphics[width=0.49\textwidth]{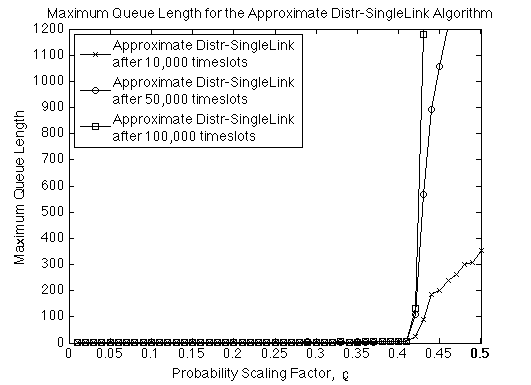}
\caption{The maximum queue lengths for the distributed algorithm Distr-SingleLink.   The problem instances are based on random topology with $n = 200$, $\ell_{\min} = 1$, $\ell_{\max} = 20$, $\alpha = 2.5$ and $\beta = 1$.} \label{fig:maxqueue}
\end{center}
\end{figure}

We implemented \alg{General} (using \alg{Distr-SingleLink}) on a random topology. As Fig.~\ref{fig:maxqueue} shows, efficiency ratios upto 0.4 is achieved, which is quite good. Details are omitted due to space restrictions.

\bibliographystyle{plain}
\bibliography{./references}		

\end{document}